\documentclass[12pt, leqno]{article}

\usepackage{amsthm} \usepackage{amsmath} \usepackage{amsfonts}
\usepackage{amssymb} \usepackage{latexsym} \usepackage{enumerate}
\usepackage[dvips]{color}
 \usepackage{mathrsfs}
\usepackage[numbers]{natbib}

\theoremstyle{plain}%
\newtheorem{Theorem}{Theorem}[section] %
\newtheorem{Lemma}[Theorem]{Lemma}
\newtheorem{Proposition}[Theorem]{Proposition} %

\theoremstyle{definition}%
\newtheorem{Assumption}[Theorem]{Assumption}%
\newtheorem{Example}[Theorem]{Example} %

\theoremstyle{remark}%
\newtheorem{Remark}[Theorem]{Remark} %

\newcommand{\set}{\triangleq}

\newcommand{\envspace}{\vspace{2mm}}

\renewcommand{\mathcal}{\mathscr}

\newcommand{\PP}{\mathbb{P}}

\renewcommand{\epsilon}{\varepsilon}

\newcommand{\Dom}{\operatorname{\mathrm{Dom}}}

\renewcommand{\P}{\PP}
\newcommand{\Pas}{\text{$\P$--a.s.}}

\usepackage{srcltx}
\numberwithin{equation}{section}
\begin{document}

\title{Utility maximization in the large markets}
\author{Oleksii Mostovyi\\
  The University of Texas at Austin,\\
  Department of Mathematics,\\
  Austin, TX 78712-0257 \\
  (mostovyi@math.utexas.edu)} \date{\today}


\maketitle
\begin{abstract}
In the large financial market, which is described by a model with { countably} many traded assets,
we formulate the problem of the expected utility maximization. Assuming that the preferences of an economic agent are modeled with
a stochastic utility 
and that the consumption occurs according to a stochastic clock, we obtain the ``usual'' conclusions of the 
utility maximization theory. 
We also give a characterization of the value function in the large market in terms of a sequence of the value functions in the finite-dimensional~models.

\vspace{2mm}
 \end{abstract}

\let\thefootnote\relax\footnote{
The author would like to thank Dmitry Kramkov, Mihai S\^irbu, and Gordan \v{Z}itkovi\'{c}
for the discussions on the topics of the paper.
This work is supported by the National Science Foundation under Grant No. DMS-0955614, PI Gordan \v{Z}itkovi\'{c}.
}

\section{Introduction}

In the mathematical finance literature, the notion of the large security market was introduced by \cite{KabKram1}
as a sequence of probability spaces with  the corresponding time horizons and the  semimartingales representing the traded assets. 
Investigation of the no-arbitrage conditions in the large market settings has naturally attracted the 
attention of the research community and is done in~\cite{KabKram2, Klein2000, Klein2003, Klein2006, KleinSchach1, KleinSchach2}, whereas 
the  questions related to completeness are considered in~\cite{
BjorkDiMasiKabRung, BjorkKabRung, DeDonno, DeDonnoPratelli2004, Taflin2005}.

In contrast to~\cite{KabKram1, KabKram2}, 
\cite{BjorkNaslund} 
assumed that a large market consists of one probability space, but the number of traded assets is countable, and among other contributions developed the arbitrage pricing theory results in such settings. 
Note that the models with countably many assets embrace the ones with the stochastic dimension of the stock price process (considered e.g. in~\cite{Strong}).  
\cite{DeDonnoGuasoniPratelli} extended the formulation in~\cite{BjorkNaslund}
 to a model driven by a sequence of semimartingales and established the standard conclusions of the theory for the utility maximization from terminal wealth problem as well as obtained  the dual characterization of the superreplicable claims. Their results are based on the notion of a stochastic integral with respect to a sequence of semimartingales from \cite{DeDonnoPratelli}. 
The Merton portfolio problem in the settings with infinitely many traded zero-coupon bonds is investigated in~\cite{EkelandTaflin2005, TehranchiRinger}.
Other applications of the large market models in the analysis of the fixed income securities are considered in~\cite{BjorkDiMasiKabRung, BjorkKabRung, CarmonaTehranchi,  CarmonaTehranchiBook, DeDonnoPratelli2004, Taflin2005}.



 We consider a market with countably many traded assets driven by a sequence of a semimartingales (as in~\cite{DeDonnoGuasoniPratelli}).
 In such settings, we formulate Merton's portfolio problem for a rational economic agent whose preferences are specified via a stochastic utility of Inada's type defined on the positive real line and whose consumption follows a stochastic clock. 
We establish the standard existence and uniqueness results for the primal and dual optimization problems under the condition of finiteness of both primal and dual value functions.
 We also  characterize the primal and dual value functions in terms of the appropriate limits of the sequences of the value functions in the finite-dimensional models.
In particular, we extend the utility maximization results in~\cite{DeDonnoGuasoniPratelli} by adding the intermediate consumption and assuming randomness
 of the agent's preferences.

The proof of our results hinges on the dual characterization of the admissible consumption processes given in Proposition~\ref{dualCharacterization}, which  allows to link the present model with the abstract theorems of~\cite{Mostovyi2011}.
Note that our formulation of the admissible consumptions and trading strategies relies on the notion of 
the stochastic integral with respect to a sequence of semimartingales in the sense of~\cite{DeDonnoPratelli}. 

We believe that our results provide a convenient set of conditions for analyzing other problems in the settings of the large markets with or without the
presence of the intermediate consumption, such as robust utility maximization, optimal investment with random endowment, utility-based pricing, and existence of equilibria.

The remainder of the paper is organized as follows. Section~\ref{sectionModel} contains the model formulation and the main results, 
which are formulated in Theorem~\ref{mainTheorem} and Lemma~\ref{secondMain}. Their proofs are given in section~\ref{sectionProofs}.

\section{The model and the main result}\label{sectionModel}
We consider a filtered probability space $\left(\Omega, \mathcal F, (\mathcal F_t)_{t\in[0,T]}, \mathbb P \right)$, where the filtration 
$(\mathcal F_t)_{t\in[0,T]}$ satisfies the usual conditions, $\mathcal F_0$ is the
completion of the trivial $\sigma$-algebra.
As in~\cite{BjorkNaslund, DeDonnoGuasoniPratelli}, we assume that there is one fixed market which consists of a riskless bond and a sequence of semimartingales $S = (S^n)_{n\geq 1} = \left((S_t^i)_{t\in[0,T]}\right)_{i= 1}^{\infty}$ that describes the evolution of the stocks. The price of the bond is supposed to be equal to $1$ at all times. 


The notion of a strategy on the large market relies on the finite-dimensional counterparts, whose definitions we specify first.
For $n\in\mathbb N$, an {\it $n$-elementary strategy} is an $\mathbb R^n$-valued, predictable process,which is integrable with respect to $(S^i)_{i\leq n}$. 
An {\it elementary strategy}  is a strategy which is $n$-elementary for some $n$.  
For $x\geq 0$, an $n$-elementary strategy $H$ is {\it $x$-admissible} if 
$H\cdot S = \sum\limits_{i\leq n}H^i\cdot S^i$ is uniformly bounded from below by the constant $-x$~\Pas 
~Let $\mathcal H^n$ denote the set of $n$-elementary strategies that are also $x$-admissible for some~$x\geq 0$.

In the present settings specification of the admissible wealth processes and trading strategies
 is based on integration with respect to a sequence of semimartingales in the sense of~\cite{DeDonnoPratelli}. 
Thus we recall several definitions from~\cite{DeDonnoPratelli},
upon which the formulation of the set of admissible consumptions is based. The reader that is familiar with this construction might proceed to 
definition of an $x$-admissible generalized strategy. 
 Recall that $\mathbb {R^N}$ is the space of all real sequences.
An {\it unbounded functional} on $\mathbb {R^N}$ is a linear functional $F$, whose domain ${\rm Dom}(F)$ is a subspace of $\mathbb {R^N}$.
A {\it simple integrand} is a finite {sum} of bounded predictable processes of the form $\sum\limits_{i\leq n}h^ie^i$, where $(e^i)$ is the 
canonical basis for $\mathbb {R^N}$ and $h^i$'s are one-dimensional bounded and predictable processes.

 A process $H$ with values in the set of unbounded functionals on $\mathbb {R^N}$ is {\it predictable} if there exists a sequence of simple integrands $(H^n)$, such that
 $
  H = \lim\limits_{n\to\infty}H^n~ \Pas,
 $
which means that $x\in\Dom(H)$ if the sequence $(H^n)$ converges and $\lim\limits_{n\to\infty}H^n(x) = H(x)$.

 A predictable process $H$ with values in the set of unbounded functionals on $\mathbb {R^N}$ is {\it integrable} with respect to $S$ if 
 there exists a sequence $(H^n)$ of simple integrands, such that $(H^n)$  converges to $H$ and the sequence of semimartingales $(H^n\cdot S)$ 
 converges to a semimartingale $Y$ in the semimartingale topology. In this case, we define the stochastic integral $H\cdot S$ to be $Y$.


 For every  $x\geq 0$, a process $H$ is an {\it $x$-admissible generalized strategy} if $H$ is integrable with respect to the semimartingale $S$ 
and there exists an approximating sequence $(H^n)$ of
$x$-admissible elementary strategies, such that $(H^n\cdot S)$ converges to $H\cdot S$ in the semimartingale topology. Note that this is Definition~2.5 
from~\cite{DeDonnoGuasoniPratelli}.

Let us define a {\it portfolio} $\Pi$ as a triple $(x, H, c),$ where the constant 
$x$ is an initial value, 
$H$ is a predictable and {admissible} $S$-integrable process (with the values in the set of unbounded functionals on $\mathbb {R^N}$)   
specifying the amount of each asset held in the portfolio, 
and
$c=\left(c_t\right)_{t\in[0,T]}$ is a nonnegative and optional process that
specifies the consumption rate in the units of the bond.

Hereafter we fix a \textit{stochastic clock} $\kappa =
\left(\kappa_t\right)_{t\in[0,T]}$, which is a
non-decreasing, c\`adl\`ag, adapted process such that
\begin{equation}
  \label{stochasticClock}
\kappa_0 = 0, ~~ \mathbb P\left[\kappa_{T}>0 \right]>0,\text{ and } \kappa_{T}\leq A
\end{equation}
for some  finite constant $A$. Stochastic clock represents the notion of time
according to which consumption occurs. {Note that, in view of the utility maximization problem (\ref{primalProblem}) defined below,
we will only consider consumption processes that are absolutely continuous with respect to $d\kappa$, i.e. of the
form $c\cdot \kappa$, since the other consumptions 
are suboptimal.}

We will use the following notation: for  arbitrary constants $x$ and $y$  and processes $X$ and $Y$, $(x+ y XY)$ denotes the process
$(x + y X_tY_t)_{t\in[0,T]}$. For a portfolio $(x, H, c)$, we define the {\it wealth process} as 
\begin{equation}\nonumber
    X = x + H\cdot S - c\cdot \kappa.
\end{equation}
Note that the closure of the sets of wealth processes  in  the semimartingale topology  is investigated in~\cite{DeDonnoGuasoniPratelli, Kardaras_emery} (with the corresponding definitions of a wealth process being different from the one here).
For $x\geq 0$, we define the set of $x$-admissible consumptions as 
\begin{equation}\nonumber
\begin{array}{rcl}
 \mathcal A(x) &\triangleq& \left\{~c\geq 0: c{\text~is~optional,~and ~there~exists} \right.\\
 &&~~{\text  ~an~}x{\text -admissible~generalized~strategy}~H,\\
 &&\left.~~{\text ~s.t.~}x + H\cdot S - c\cdot\kappa\geq 0\right\}.\\
\end{array}
\end{equation}
Thus a constant strictly positive consumption
$c^{*}_t \triangleq x/A, ~t\in[0,T]$, belongs to $\mathcal A(x)$ for every $x>0$. 



For $n\geq 1$, let $\mathcal Z^n$ denote the set of c\`adl\`ag densities of equivalent martingale measure for $n$-elementary strategies, i.e.
\begin{displaymath}
\begin{array}{rcl}
\mathcal{Z}^n&\set&
  \left\{Z>0:~ Z {\rm~is~ a~ c\grave{a}dl\grave{a}g~martingale,~s.t.~}Z_0=1~{\rm and}\right.\\
&&\left. ~~(1 + H\cdot S)Z{\rm~is~a~local~ martingale~for~every~}H\in\mathcal H^n, \right.\\
&&\left.{~~H~is~1-admissible}\right\}.
\end{array}
\end{displaymath}
 Note that 
$\mathcal{Z}^{n+1} \subseteq \mathcal{Z}^n $, $n\geq 1$. We also define
\begin{displaymath}
 \mathcal Z \triangleq \bigcap\limits_{n\geq 1}\mathcal Z^n,
\end{displaymath}
and assume that
\begin{equation}\label{ZisNotEmpty}
  \mathcal{Z} \neq \emptyset,
\end{equation}
which coincides with the no-arbitrage condition in~\cite{DeDonnoGuasoniPratelli}. 

The preferences of an economic agent are modeled via a stochastic utility $U:[0,T]\times\Omega\times [0,\infty)\to \mathbb R\cup\{-\infty\}$ that satisfies the conditions below.
\begin{Assumption}
  \label{Assumption1}
  For every $(t, \omega)\in[0, T]\times\Omega$ the function $x\to
  U(t, \omega, x)$ is strictly concave, increasing, continuously
  differentiable on $(0,\infty)$ and satisfies the Inada conditions:
  \begin{equation}\nonumber
    \lim\limits_{x\downarrow 0}U'(t, \omega, x) =
    +\infty \quad \text{and} \quad \lim\limits_{x\to
      \infty}U'(t, \omega, x) \set 0,
  \end{equation}
  where $U'$ denotes the partial derivative with respect to the third argument.
 At $x=0$ we suppose, by continuity, $U(t, \omega, 0) = \lim\limits_{x \downarrow 0}U(t, \omega, x)$, this value
 may be $-\infty$.
For every
  $x\geq 0$ the stochastic process $U\left( \cdot, \cdot, x \right)$ is
  optional.
\end{Assumption}
{\looseness +1}The conditions on $U$ coincide with the ones in \cite{Mostovyi2011} (on the finite time horizon). For simplicity of notations for a nonnegative optional  process $c$, the processes with trajectories   
$\left(U(t,\omega, c_t(\omega))\right)_{t\in[0,T]}$, 
 $\left(U'(t,\omega, c_t(\omega))\right)_{t\in[0,T]}$, and $\left(U^{-}(t,\omega, c_t(\omega))\right)_{t\in[0,T]}$ (where $U^{-}$ designates the negative part of $U$)
will be denoted by $U(c)$, $U'(c)$, and $U^{-}(c)$ respectively.

For a given initial capital $x>0$ the goal of the agent is to maximize
his expected utility. The value function of this problem is denoted by
\begin{equation}\label{primalProblem}
  u(x) \set \sup\limits_{c\in\mathcal{A}(x)} \mathbb{E}\left[
    U(c)\cdot\kappa_T
  \right],\quad x>0.
\end{equation}
We use the convention
\begin{equation}\label{wellDefined_u}
  \mathbb{E}\left[
     U(c)\cdot\kappa_T\right] \set -\infty
  \quad \text{if} \quad \mathbb{E}\left[ 
   U^{-}(c)\cdot\kappa_T
  \right]= +\infty.
\end{equation}

To study~\eqref{primalProblem} we employ standard duality arguments as
in \cite{KS} and \cite{Zitkovic} and define the \textit{conjugate stochastic
field} $V$ to $U$ as
\begin{equation}
  \nonumber
  V(t,\omega, y) \set \sup\limits_{x>0}\left( U(t,\omega, x) - xy \right)
  ,\quad
  \left(t, \omega, y \right) \in[0, T]\times\Omega\times[0,
  \infty).
\end{equation}
It is well-known that $-V$ satisfies Assumption~\ref{Assumption1}. For $y\geq 0$, we
also denote
\begin{equation}\nonumber
\begin{array}{c}
\hspace{-15mm}{\mathcal Y}(y) \set {\rm cl}\left\{Y: Y{\rm~is~c\grave adl\grave ag~adapted~and }
\right.\\
\hspace{35mm}\left. 0\leq Y\leq yZ ~\left(d\kappa\times\mathbb
    P\right){\rm~a.e.~for~some~}Z\in{\mathcal Z} \right\},
\end{array}
\end{equation}
\looseness+1 where the closure is taken in the topology of convergence in measure $\left(d\kappa\times\mathbb
    P\right)$ on the space of finite-valued optional processes. We will denote this space $\mathbb L^0\left(d\kappa\times\mathbb
    P\right)$ or $\mathbb L^0$ for brevity. 

{\looseness+1} 
Similarly to composition of $U$ with $c$, for a nonnegative optional  process $Y$, the stochastic processes, whose realizations are 
$\left(V(t,\omega, Y_t(\omega))\right)_{t\in[0,T]}$ and  $\left(V^{+}(t,\omega, Y_t(\omega))\right)_{t\in[0,T]}$ (where $V^{+}$ is the positive part of $V$),
will be denoted by $V(Y)$ and $V^{+}(Y)$ respectively.
After these preparations, we define the value function of the dual optimization
problem~as
\begin{equation}\label{dualProblem}
  v(y) \set \inf\limits_{Y\in{\mathcal Y}(y)} \mathbb{E}\left[
     V(Y)\cdot\kappa_T
  \right],\quad y>0,
\end{equation}
where we use the convention:
\begin{equation}\label{wellDefined_v}
  \mathbb{E}\left[ 
        V(Y)\cdot\kappa_T
   \right] \set +\infty
  \quad \text{if} \quad 
   \mathbb{E}\left[ 
       V^{+}(Y)\cdot\kappa_T
  \right] = +\infty.
\end{equation} 

The following theorem constitutes the main contribution of the present article.
\begin{Theorem}\label{mainTheorem}
  Assume that conditions~(\ref{stochasticClock}) and (\ref{ZisNotEmpty}) and
  Assumption \ref{Assumption1} hold true and suppose
  \begin{equation}\nonumber
    v(y)<\infty~~ for~ all~y>0~~~ and~~~
    u(x) > -\infty~~ for~ all~x>0.
  \end{equation}
  Then we have:
  \begin{enumerate}
  \item $u(x)< \infty$ for all $x>0,$ $v(y)>-\infty$ for all $y>0.$
    The functions $u$ and $v$ are conjugate, i.e.,
    \begin{equation}\nonumber
      \begin{array}{rcl}
        v(y) &=& \sup\limits_{x>0}\left(u(x) - xy\right),\quad y>0,\\
        u(x) &=& \inf\limits_{y>0}\left(v(y) + xy\right),\quad x>0.\\
      \end{array}
    \end{equation}
    The functions $u$ and $-v$ are continuously differentiable on $(0,
    \infty),$ strictly increasing, strictly concave and satisfy the
    Inada conditions:
    \begin{equation}\nonumber
      \begin{array}{lcr}
        u'(0) \set \lim\limits_{x\downarrow 0}u'(x) = +\infty, && -v'(0)
        \set \lim\limits_{y\downarrow 0}-v'(y) = +\infty,\\
        u'(\infty) \set \lim\limits_{x\to\infty}u'(x) = 0, && -v'(\infty) \set
        \lim\limits_{y\to\infty}-v'(y) = 0.\\
      \end{array}
    \end{equation}

  \item For every $x>0$ and $y>0$ the optimal solutions $\hat{c}(x)$ to
    (\ref{primalProblem}) and $\hat{Y}(y)$ to (\ref{dualProblem})
    exist and are unique. Moreover, if $y=u'(x)$ we have the dual
    relations
    \begin{displaymath}
       \hat Y(y) = U'(\hat c(x)),\quad (d\kappa\times \mathbb P)~a.e.
        \end{displaymath}
    and
    \begin{displaymath}
      \mathbb{E}\left[ 
          \left((\hat{c}(x)\hat{Y}(y))\cdot \kappa\right)_T
        \right]=xy.
    \end{displaymath}

\item We have,
\begin{displaymath}
\begin{array}{rclc}
      v(y)& =& \inf\limits_{Z\in\mathcal{Z}}\mathbb{E}\left[ 
         V(yZ)\cdot\kappa_T
        \right],& y>0,\vspace{2mm}\\
\end{array}
\end{displaymath}
\end{enumerate}

\end{Theorem}

\subsection{Large market as a limit of a sequence of finite-dimensional markets}
Motivated by the question of liquidity, we discuss the convergence of the value functions as the number of available traded securities increases.
For this purpose, we need the following definitions. 
For every $n\geq 1$, we set
\begin{equation}\nonumber
\begin{array}{rcl}
 \mathcal A^n(x)&\triangleq& \left\{ {\text optional}~ c\geq 0: {\text ~there ~exists~}H\in\mathcal H^n{\text~s.t.~}\right.\\
 &&~~\left. 
 x + H\cdot S_T - c\cdot\kappa_T\geq 0~\Pas\right\}, \\
 \end{array}
\end{equation}
\begin{equation}\label{primalProblemN}
  u^n(x) \set \sup\limits_{c\in\mathcal{A}^n(x)} \mathbb{E}\left[
     U(c)\cdot \kappa_T
  \right],\quad x>0,
\end{equation}
\begin{equation}\nonumber
\begin{array}{rcl}
{\mathcal Y}^n(y) &\triangleq &{\rm cl}\left\{Y: Y{\rm~is~c\grave adl\grave ag~adapted~and }
\right.\\
&&\quad ~~\left. 0\leq Y\leq yZ ~\left(d\kappa\times\mathbb
    P\right){\rm~a.e.~for~some~}Z\in{\mathcal Z}^n \right\},
\end{array}
\end{equation}
where the closure is taken in $\mathbb L^0$,
\begin{equation}\label{dualProblemN}
  v^n(y) \set \inf\limits_{Y\in{\mathcal Y}^n(y)} \mathbb{E}\left[
       V(Y)\cdot \kappa_T
  \right],\quad y>0,
\end{equation}
and assume the conventions (\ref{wellDefined_u}) and (\ref{wellDefined_v}). Note that for every 
$z>0$, both $(u^n(z))$ and $(v^n(z))$ are increasing sequences. 
We suppose that
\begin{equation}\label{closureA}
\mathcal A(1-\varepsilon) \subset {\rm cl}\left(\bigcup\limits_{n\geq 1}\mathcal A^n(1)\right)\quad
{\rm for~every}\quad \varepsilon\in(0,1],
\end{equation}
where the closure is taken in $\mathbb L^0$.

Let $1_E$ denotes the indicator function of a set $E$. 
\begin{Remark}\label{keyRemark}
It follows from Proposition~\ref{dualCharacterization} below and Fatou's lemma that ${\rm cl}\left(\bigcup\limits_{n\geq 1}\mathcal A^n(1)\right) \subseteq \mathcal A(1)$.  Assumption (\ref{closureA}) gives a weaker version of the reverse inclusion. Note that (\ref{closureA})
 holds if either of the conditions below~is~valid.
\begin{enumerate}
\item
 $\kappa_t = 1_{T}(t)$, $t\in[0,T]$,
  i.e. if (\ref{primalProblem}) defines the problem of optimal investment from terminal wealth. Then (\ref{closureA}) follows from    
Lemma 3.4~in~\cite{DeDonnoGuasoniPratelli}.

\item The process $S$ is (componentwise) continuous. This is the subject of Lemma~\ref{3-16-1} below.
\end{enumerate}
\end{Remark}
\begin{Lemma}\label{secondMain}
Assume that there exists $n\in\mathbb N$, such that 
\begin{equation}\label{secondFiniteness}
u^n(x)>-\infty\quad {\text for~ every}\quad x>0,\quad v(y)<+\infty\quad{\text for~ every}\quad y>0.
\end{equation}
Then, under conditions (\ref{stochasticClock}), (\ref{ZisNotEmpty}), and (\ref{closureA}) as well as 
  Assumptions~\ref{Assumption1}, we have 
\begin{equation}\label{9-18-2} u(x) = \lim\limits_{n\to\infty}u^n(x),
\quad x>0,\quad {\text and}\quad  v(y)= \lim\limits_{n\to\infty}v^n(y),
\quad y>0.
\end{equation}
\end{Lemma}
{
\begin{Remark}
(\ref{secondFiniteness}) imply finiteness of $v$, $-u$,  $v^n$, and $-u^n$, $n\geq 1$, that are also convex.
Theorem 3.1.4 in~\cite{LemHur} ensures that convergence in (\ref{9-18-2}) is uniform on compact subsets of $(0, \infty)$.
Moreover, Theorem 25.7 in \cite{Rok} asserts that the derivatives $(v^n)'$ and $(u^n)'$, $n\geq 1$, also converge uniformly on compact intervals in 
$(0,\infty)$ to $v'$ and $u'$, respectively.
\end{Remark}
}

{
Lemma~\ref{secondMain} shows that the value function in the market with countably many assets is the limit of
the value functions of the finite dimensional models.
The following example shows that the optimal portfolio in 
the market with infinitely many traded assets is not a limit of the optimal portfolios in the finite dimensional markets, in general. 
The important technical feature in the construction of this example, is that
in  each finite dimensional market {\it the last stock has the biggest expected return}.  
}

\begin{Example}
{We consider a one-period model, where there is a riskless bond with $S^0\equiv 1$, and a sequence of stocks $(S^i)$, such 
that $ S^i_0 = 1$ for every $i$ and $(S^i_1)$  are independent random variables taking values in $\{\tfrac{1}{2}, 2\}$ with probabilities
$1-p_i$ and  $p_i$ respectively, where $(p_i)$ is an {\it increasing} sequence. Therefore, we have
$$\max\limits_{k\in\{1,\dots,n\}} \mathbb E\left[ S^k_1\right] = \mathbb E\left[ S^n_1\right],\quad n\geq 1,$$
i.e. the last stock of each finite dimensional market has the greatest expected return. 
Note that (\ref{ZisNotEmpty}) holds.
}

{
We assume that the preferences of an economic agent 
are specified by a bounded utility function $U$ defined on the positive real line that is strictly increasing, strictly concave, continuously differentiable and 
satisfies the Inada conditions.  Let the stochastic clock $\kappa$ corresponds to the problem of utility maximization of terminal wealth.
Then (\ref{closureA}) holds by the first item of Remark~\ref{keyRemark}, whereas boundedness of $U$ implies (\ref{secondFiniteness}).
Therefore, the assertions of Lemma~\ref{secondMain} hold. We also impose the following technical assumption  
\begin{equation}\label{9-18-1}
p_1> \frac{U\left(1\right) - U\left(\tfrac{1}{2}\right)}{U\left(2\right) - U\left(\tfrac{1}{2}\right)}\vee \frac{1}{3},
\end{equation}
which in particular implies that 
\begin{equation}\label{9-18-5}
U(1) = 
\mathbb E\left[U(S^0_1) \right] < \mathbb E\left[ U(S^1_1)\right].
\end{equation} 
}

{
For simplicity of notations, we will assume that the initial wealth of the agent equals to $1$.
Let $h^N_i$ be the optimal number of shares of the $i$-th asset in the market, where $N$ stocks are available for trading, $N\geq 1$.
Admissibility condition implies that $h^N_0 \geq 0$, i.e. the number of shares of the riskless asset must be nonnegative.
Monotonicity of $(p_i)$ results in the following inequalities
\begin{equation}\label{9-18-3}
h^N_1\leq h^N_2\leq \dots\leq h^N_N,\quad N\geq 1.
\end{equation}
It follows from convexity and monotonicity of $U$ as well as  (\ref{9-18-1}) that $h^N_i \geq 0$ (if, by contradiction, $h^N_i<0$, a portfolio with 
$0$ units of $i$-th stock and $h^N_0 + h^N_i$ units of the riskless asset is admissible, it corresponds to the same initial wealth and gives a higher value of the expected utility).
Nonnegativity of $h^N_i$'s and (\ref{9-18-3}) gives
\begin{displaymath}
h^N_i \leq \frac{1}{N-i +1},\quad i = 1,\dots,N,\quad N\geq 1,
\end{displaymath}  
This implies that 
\begin{equation}\label{9-18-4}
\lim\limits_{N\to\infty} h^N_i = 0,\quad i\geq 1\textsc{}.
\end{equation}
Consequently, in the market with countably many stocks, a portfolio that 
is the limit of the optimal finite dimensional portfolios (i.e. satisfies (\ref{9-18-4}))
can  have nontrivial allocation only in the riskless asset. This gives the value of the expected utility $U(1)$.  In view of (\ref{9-18-5}),
such a portfolio is suboptimal.
}
\end{Example} 

\section{Proofs}\label{sectionProofs}
In the core of the proof of Theorem~\ref{mainTheorem} lies the following result. 
 \begin{Proposition}\label{dualCharacterization}
 Let conditions  (\ref{stochasticClock}) and (\ref{ZisNotEmpty}) hold. Then 
a nonnegative optional process 
    $c$ belongs to $\mathcal A(1)$ if and only if
 \begin{equation}\label{dualChar}
  \sup_{Z\in\mathcal Z}\mathbb E\left[ 
    \left((cZ)\cdot \kappa\right)_T  
  \right] \leq 1.
 \end{equation} 
 \end{Proposition}
The proof of Proposition~\ref{dualCharacterization} will be given via several lemmas. 
\begin{Lemma}\label{SholiPositive}
Let $H$ be a 1-admissible generalized integrand. Under the conditions 
 Proposition~\ref{dualCharacterization}, $X \triangleq 1 + H\cdot S$ is nonnegative \Pas ~and  
for every $Z\in\mathcal Z$, $ZX$ is a supermartingale.
\end{Lemma}
{
The proof of Lemma~\ref{SholiPositive} is straightforward, it is therefore skipped. Note that discussion of the second assertion
of the lemma is presented on p. 2011 of~\cite{DeDonnoGuasoniPratelli}.
}
\begin{Lemma}\label{SholiBounded}
 Let $H$ be a $1$-admissible  generalized  strategy, $c$ be a nonnegative optional process. Under the conditions 
 Proposition~\ref{dualCharacterization}, the following statements are equivalent
 \begin{enumerate}[(i)]\item
 $$
  c\cdot\kappa_T\leq 1 + H\cdot S_T, \quad \Pas,
 $$
 \item
 \begin{equation}\nonumber
\begin{array}{c}
  c\cdot \kappa \leq 1 + H\cdot S,\quad \Pas\\
  ({\text i.e.}\quad c\cdot\kappa_t \leq 1+H\cdot S_t \quad {\text for~every~}t\in[0,T], \quad \Pas).\\
\end{array}
 \end{equation}
 \end{enumerate}
\end{Lemma}
\begin{proof}
 Let us assume that $(i)$ holds and fix $Z\in\mathcal Z$. It follows from Lemma~\ref{SholiPositive} that $Z(1 + H\cdot S)$ is a supermartingale. Therefore, using monotonicity of $c\cdot\kappa$,
 for every $t\leq T$ we have
 \begin{displaymath}
 \begin{array}{c}
  Z_t (c\cdot\kappa_t)= \mathbb E\left[ Z_T(c\cdot\kappa_t)|\mathcal F_t\right]\leq \mathbb E\left[ Z_T(c\cdot\kappa_T)|\mathcal F_t\right] \\
  \leq  \mathbb E\left[Z_T (1+H\cdot S_T)|\mathcal F_t\right] \leq Z_t(1+H\cdot S_t),\\
  \end{array}
 \end{displaymath}
 which implies $(ii)$ in view of the {\it strict} positivity of $Z$ and the right-continuity of both $(1+H\cdot S)$ and $(c\cdot\kappa)$,
 where the latter follows e.g. from Proposition I.3.5 in~\cite{Jahod-Shiryaev}.

\end{proof}

\begin{proof}[Proof of Proposition~\ref{dualCharacterization}]

 Let $c\in\mathcal A(1)$. 
 Fix $Z\in\mathcal Z$ and $T>0$. Then 
there exists a $1$-admissible generalized strategy $H$, such that
\begin{displaymath}
 1 + H\cdot S_T \geq c\cdot\kappa_T.
\end{displaymath}
Multiplying both sides by $Z$ and taking the expectation, we get
\begin{equation}\label{12-18-1}
 \mathbb E\left[ Z_T(1 + H\cdot S_T)\right] \geq \mathbb E\left[ Z_T(c\cdot\kappa_T)\right],
\end{equation}
where the right-hand side (via monotonicity of $ c\cdot\kappa$
and an application of Theorem I.4.49 in \cite{Jahod-Shiryaev}) can be rewritten as
\begin{equation}\label{12-18-2}
 \mathbb E\left[ Z_T (c\cdot\kappa_T)\right] = \mathbb E\left[ ((Zc)\cdot\kappa)_T\right].
\end{equation}

By definition of $H$, there exists a sequence  $(H^n)$ of $1$-admissible elementary strategies, such that
\begin{displaymath}
 (H^n\cdot S)_{n\geq 1} \quad{\rm converges~to}\quad H\cdot S \quad{\rm in~the~semimartingale~topology}.
\end{displaymath}
Consequently, $(H^n\cdot S_T)$ converges to $H\cdot S_T$ in probability, and therefore there exist a subsequence, which we
still denote $(H^n\cdot S)$, such that $(H^n\cdot S_T)$ converges to $H\cdot S_T$ \Pas~Therefore, for every $Z\in\mathcal Z$ we 
obtain from  the definition of $1$-admissibility and Fatou's lemma 
\begin{displaymath}
 1\geq \liminf\limits_{n\to\infty}\mathbb E\left[ Z_T(1+ H^n\cdot S_T)\right] \geq \mathbb E\left[ Z_T(1+ H\cdot S_T)\right].
\end{displaymath}
Combining this with (\ref{12-18-1}) and (\ref{12-18-2}), we conclude that
\begin{displaymath}
 1\geq \mathbb E\left[ (( Zc)\cdot\kappa)_T\right],
\end{displaymath}
   which holds for every ${Z\in\mathcal Z}$.

Conversely, let (\ref{dualChar}) holds.  Using the same argument as in (\ref{12-18-2}), we obtain from (\ref{dualChar}) that
\begin{displaymath}
 1\geq \sup_{Z\in\mathcal Z}\mathbb E\left[ Z_T(c\cdot\kappa)_T\right].
\end{displaymath}
Consequently, the random variable $ c\cdot \kappa_T$ satisfies the assumption (i) of Theorem 3.1 in~\cite{DeDonnoGuasoniPratelli} with $x =1$.
Therefore, we obtain from this theorem that there exists a  $1$-admissible generalized strategy $H$ such that
\begin{displaymath}
 c\cdot\kappa_T \leq 1 +  H\cdot  S_T.
\end{displaymath}
By Lemma~\ref{SholiBounded}, this implies that $c\in\mathcal A(1)$.
This concludes the proof of the proposition.
\end{proof}
Let $\mathbb L^0_{+}$ denote the positive orthant of $\mathbb L^0$. We recall that a subset $A$ of $\mathbb L^0_{+}$ is called {\it solid} if $f\in A$, $g\in\mathbb L^0_{+}$, and 
$g\leq f$ implies that $g\in A$, a subset  $B\subset\mathbb L^0_{+}$ is the {\it polar} of $A$, if $B = \left\{h\in\mathbb L^0_{+}:~\mathbb E\left[ ((hf)\cdot\kappa)_T\right]\leq 1, {\rm ~for~every~}f\in 
A \right\}$, in this case we denote $B = A^{o}$.
\begin{Lemma}\label{bipolarAY}
Under the conditions of Proposition~\ref{dualCharacterization}, we have
\begin{enumerate}[(i)]
\item The sets $\mathcal A(1)$ and $\mathcal Y(1)$ 
are convex, solid, and closed subsets of $\mathbb L^0$.
\item $\mathcal A(1)$ and $\mathcal Y(1)$ satisfy the bipolar relations
\begin{displaymath}
\begin{array}{rclcl}
c\in\mathcal A(1)&{\Leftrightarrow}&\mathbb E\left[((cY)\cdot\kappa)_T \right]\leq 1,&{\rm for~every}&Y\in\mathcal Y(1),\\
Y\in\mathcal Y(1)&{\Leftrightarrow}&\mathbb E\left[((cY)\cdot\kappa)_T \right]\leq 1,&{\rm for~every}&Y\in\mathcal A(1).\\
\end{array}
\end{displaymath}
\item Both 
$\mathcal A(1)$ and $\mathcal Y(1)$ contain strictly positive elements.
\end{enumerate}
\end{Lemma}
\begin{proof}
 
 Assertions of item $(iii)$ follow from conditions (\ref{stochasticClock}) and (\ref{ZisNotEmpty}) respectively.
Now in view of Proposition~\ref{dualCharacterization}, the proof of the remaining items goes along the lines of the proof of Proposition~4.4 in~\cite{Mostovyi2011}. It is therefore omitted~here.
\end{proof}

\begin{Lemma}\label{propertiesOfZ}
Under the conditions of Proposition~\ref{dualCharacterization}, we have
\begin{enumerate}[(i)]
\item  
$\sup\limits_{Z\in\mathcal Z} \mathbb E\left[((cZ)\cdot\kappa)_T \right] 
= \sup\limits_{Y\in\mathcal Y(1)} \mathbb E\left[((cY)\cdot\kappa)_T \right]$ 
for every $c\in\mathcal A(1)$,
\item the set $\mathcal Z$ is closed under
 the countable convex combinations, i.e. for every sequence  $(Z^m)$ in $\mathcal Z$ and a sequence of positive numbers $(a^m)$ 
such that $\sum\limits_{m\geq 1}a^m = 1$, the process $Z \triangleq \sum\limits_{m\geq 1}a^mZ^m$ belongs to $\mathcal Z$.
\end{enumerate}
\end{Lemma}
\begin{proof}
  For every $n\geq 1$, and $H\in\mathcal H^n$, in view of the positivity of $X\triangleq x + H^n\cdot S$ (for an appropriate $x\geq 0$),
\begin{displaymath}
\tau^k \triangleq \inf\left\{ t>0: ~X_t > k\right\}\wedge T, \quad k\geq 1,
\end{displaymath}
is a localizing sequence for $XZ$  {\it for every} $Z\in\mathcal Z$. 
This implies $(ii)$, 
whereas $(i)$ results  from Fatou's lemma and the definitions of the sets $\mathcal Z$ and $\mathcal Y(1)$.

\end{proof}
\begin{proof}[Proof of Theorem~\ref{mainTheorem}]
 By Lemma~\ref{bipolarAY}, the sets $\mathcal A(1)$ and $\mathcal Y(1)$ satisfy the assumptions of 
 Theorem 3.2 in~\cite{Mostovyi2011} that implies the assertions $(i)$ and $(ii)$ of Theorem~\ref{mainTheorem}. The 
 conclusions of item $(iii)$ supervene from Lemma~\ref{propertiesOfZ} and Theorem 3.3 in~\cite{Mostovyi2011}. This completes the proof of Theorem~\ref{mainTheorem}.
 
 \end{proof}
For the proof of Lemma~\ref{secondMain}, we need the following technical result.
\begin{Lemma}\label{lemma2-7-1}
Under the conditions of Lemma~\ref{secondMain}, for every $\varepsilon\in(0,1)$~we~have
\begin{displaymath}
\bigcap\limits_{n\geq 1}\mathcal Y^n(1) \subset \mathcal Y\left(\frac{1}{1-\varepsilon}\right).
\end{displaymath}
\end{Lemma}
\begin{proof}
Observe that by Proposition  4.4 in~\cite{Mostovyi2011}, for every $n\geq 1$, the sets $\mathcal A^n(1)$ and $\mathcal Y^n(1)$ satisfy the
bipolar relations, likewise by Lemma~\ref{bipolarAY}, we have $\mathcal A(1)^{o} = \mathcal Y(1)$.
Fix an $\varepsilon\in(0,1)$.
From (\ref{closureA})
using Fatou's lemma we~obtain
$$\mathcal A(1-\varepsilon)^{o} \supset \left(\bigcup\limits_{n\geq 1}\mathcal A^n(1)\right)^{o}.$$ Therefore we conclude
\begin{displaymath}
\mathcal Y\left(\frac{1}{1-\varepsilon}\right) = \mathcal A(1-\varepsilon)^{o} 
\supset  \left(\bigcup\limits_{n\geq 1}\mathcal A^n(1)\right)^{o}  = 
\bigcap\limits_{n\geq 1}\mathcal A^n(1)^{o}  = \bigcap\limits_{n\geq 1} \mathcal Y^n(1).
\end{displaymath}
This concludes the proof of the lemma.
\end{proof}
\begin{proof}[Proof of Lemma~\ref{secondMain}]
Without loss of generality, we will assume that $u^1(x)>-\infty,~x>0$.
We will only show the second assertion, as the proof of the first one is entirely similar.
 Also, for convenience of notations, we will assume that $y=1$. 
 Let $Z^n$ be a minimizer to the dual problem (\ref{dualProblemN}), $n\geq 1$, where
 the existence of the solutions to~(\ref{dualProblemN}) follows from Theorem 2.3 in~\cite{Mostovyi2011}. 
 
 It follows from (\ref{stochasticClock}) that the set $\mathcal Z^1$ {is} 
 bounded in $\mathbb L^1\left(d\kappa\times \mathbb P\right)$. 
 This in particular implies that
 $\mathcal Y^1(1)$ is bounded in $\mathbb L^0\left(d\kappa\times \mathbb P\right)$. Therefore, by Lemma A1.1 in~\cite{DS}, 
 there exists a sequence $\widetilde Z^n\in{\rm conv}\left( Z^n, Z^{n+1}, \dots\right)$, $n\geq 1$,  and an
 element $Z\in\mathbb L^0\left(d\kappa\times \mathbb P\right)$, such 
 that $(\widetilde Z^n)$ converges to $Z$ $\left(d\kappa\times \mathbb P\right)$-a.e. We also have
 $$Z = \lim\limits_{n\to\infty}\widetilde Z^n \in\bigcap\limits_{n\geq 1}\mathcal Y^n(1) \subset \mathcal Y\left(\frac{1}{1-\varepsilon}\right)\quad {\rm for~ every}\quad \varepsilon\in(0,1),$$
where the latter inclusion follows from Lemma~\ref{lemma2-7-1}.
 By convexity of $V$, we~get
 \begin{equation}\label{1-10-1}
  \limsup\limits_{n\to\infty} \mathbb E\left[V(\widetilde Z^n)\cdot\kappa_T \right] \leq 
  \lim\limits_{n\to\infty}v^n(1).
 \end{equation}
Note that $(\widetilde Z^n)\subset \mathcal Y^1(1)$. Consequently, using Lemma 3.5 in~\cite{Mostovyi2011}, we conclude that 
$\left(V^{-}\left(\widetilde Z^n \right)\right)$ in uniformly integrable 
(here $V^{-}$ denotes the negative part of the stochastic field $V$).
Therefore, from Fatou's lemma and (\ref{1-10-1})~we~deduce 
\begin{displaymath}
 v\left(\frac{1}{1-\varepsilon}\right) \leq \mathbb E\left[V(Z)\cdot\kappa_T \right] \leq \liminf\limits_{n\to\infty} \mathbb E\left[V(\widetilde Z^n)\cdot\kappa_T \right] \leq 
  \lim\limits_{n\to\infty}v^n(1)
\end{displaymath}
for every $\varepsilon \in(0,1)$. Taking the limit as $\varepsilon\downarrow 0$ and using the continuity of $v$ (by convexity, see Theorem~\ref{mainTheorem}), we obtain that
$$v(1) \leq \lim\limits_{n\to\infty}v^n(1).$$
Also, since $\mathcal Y(1) \subseteq \mathcal Y^n(1)$ for every $n\geq 1$, we have
$$v(1) \geq \lim\limits_{n\to\infty}v^n(1).$$
Thus, $v(1) = \lim\limits_{n\to\infty}v^n(1)$. The proof of the lemma is now complete.
\end{proof}

\begin{Lemma}\label{3-16-1}
Let $S$ be a continuous process (i.e. every component of $S$ is continuous) that satisfy (\ref{ZisNotEmpty}). Then, under (\ref{stochasticClock}), (\ref{closureA}) holds.
\end{Lemma}
\begin{proof}
Fix an $\varepsilon\in(0,1]$ and $c\in\mathcal A(1-\varepsilon)$. Let $H$ be a $(1-\varepsilon)$-admissible generalized strategy, such that
\begin{displaymath}
c\cdot \kappa\leq 1 - \varepsilon + H\cdot S,\quad \Pas
\end{displaymath}
Let $(H^n)$ be a sequence of $(1-\varepsilon)$-admissible elementary strategies, such that
$H^n\cdot S$ converges to $H\cdot S$ in the semimartingale topology. Let us define a sequence of stopping times as
\begin{displaymath}
\tau_n \triangleq \inf\left\{t\in[0,T]:~c\cdot\kappa_t{ > }
 1 + H^n\cdot S_t\right\}\wedge {(T+1)}.
\end{displaymath}
Then we have
\begin{displaymath}
\begin{array}{rcl}
\mathbb P[\tau_n {\leq} T] &\leq& \mathbb P\left[\sup\limits_{t\in[0,T]}(c\cdot\kappa_t - 1 +\varepsilon - H^n\cdot S_t) \geq \varepsilon\right] \\
&\leq &  \mathbb P\left[\sup\limits_{t\in[0,T]}(H\cdot S_t - H^n\cdot S_t) \geq \varepsilon\right],\\
\end{array}
\end{displaymath}
which converges to $0$ as $n\to\infty$. 
Let us define a sequence of consumptions $(c^n)$ as follows
\begin{displaymath}
c^n_t \triangleq c_t1_{[0, \tau_n{)}}(t),\quad t\in[0,T], \quad n\geq 1.  
\end{displaymath}
Then, by continuity of $S$ we get
\begin{displaymath}
c^n\cdot\kappa \leq 1 + H^n\cdot S\quad{\rm on~[0,\tau_n]}\quad \Pas,\quad n\geq 1.
\end{displaymath}
Since $H^n1_{[0,\tau_n]}$ is a $1$-admissible elementary strategy, we deduce that 
 $c^n\in \mathcal A^n(1)$, $n\geq 1$. One can also see that $(c^n)$ converges to $c$ in $\mathbb L^0$.  
This concludes the proof of the lemma. 
\end{proof}

\bibliographystyle{plainnat} \bibliography{finance}

\begin{thebibliography}{28}
\providecommand{\natexlab}[1]{#1}
\providecommand{\url}[1]{\texttt{#1}}
\expandafter\ifx\csname urlstyle\endcsname\relax
  \providecommand{\doi}[1]{doi: #1}\else
  \providecommand{\doi}{doi: \begingroup \urlstyle{rm}\Url}\fi

\bibitem[Bj\"ork and N\"aslund(1998)]{BjorkNaslund}
T.~Bj\"ork and B.~N\"aslund.
\newblock Diversified portfolios in continuous time.
\newblock \emph{Europ. Fin. Rev.}, 1:\penalty0 361--–387, 1998.

\bibitem[Bj\"ork et~al.(1997{\natexlab{a}})Bj\"ork, Di~Masi, Kabanov, and
  Runggaldier]{BjorkDiMasiKabRung}
T.~Bj\"ork, G.~Di~Masi, Y.~Kabanov, and W.~Runggaldier.
\newblock Towards a general theory of bond markets.
\newblock \emph{Finance Stoch.}, 1:\penalty0 141–--174, 1997{\natexlab{a}}.

\bibitem[Bj\"ork et~al.(1997{\natexlab{b}})Bj\"ork, Kabanov, and
  Runggaldier]{BjorkKabRung}
T.~Bj\"ork, Y.~Kabanov, and W.~Runggaldier.
\newblock Bond market structure in the presence of marked point processes.
\newblock \emph{Math. Finance}, 7\penalty0 (2):\penalty0 211--–239,
  1997{\natexlab{b}}.

\bibitem[Carmona and Tehranchi(2004)]{CarmonaTehranchi}
R.~Carmona and M.~Tehranchi.
\newblock A characterization of hedging portfolios for interest rate contingent
  claims.
\newblock \emph{Ann. Appl. Probab.}, 14\penalty0 (3):\penalty0 1267--1294,
  2004.

\bibitem[Carmona and Tehranchi(2006)]{CarmonaTehranchiBook}
R.~Carmona and M.~Tehranchi.
\newblock \emph{Interest Rate Models: an Infinite Dimensional Stochastic
  Analysis Perspective}.
\newblock Springer, 2006.

\bibitem[De~Donno(2004)]{DeDonno}
M.~De~Donno.
\newblock A note on completeness in large financial markets.
\newblock \emph{Math. Finance}, 14\penalty0 (2):\penalty0 295–--315, 2004.

\bibitem[De~Donno and Pratelli(2004)]{DeDonnoPratelli2004}
M.~De~Donno and M.~Pratelli.
\newblock On the use of measure-valued strategies in bond markets.
\newblock \emph{Finance Stoch.}, 8:\penalty0 87–--109, 2004.

\bibitem[De~Donno and Pratelli(2006)]{DeDonnoPratelli}
M.~De~Donno and M.~Pratelli.
\newblock Stochastic integration with respect to a sequence of semimartingales.
\newblock \emph{In Memoriam Paul-Andr\' e Meyer, S\'eminaire de Probabilit\'es
  XXXIX}, pages 119–--135, 2006.

\bibitem[De~Donno et~al.(2005)De~Donno, Guasoni, and
  Pratelli]{DeDonnoGuasoniPratelli}
M.~De~Donno, P.~Guasoni, and M.~Pratelli.
\newblock Super-replication and utility maximization in large financial
  markets.
\newblock \emph{Stochastic Process. Appl.}, 115:\penalty0 2006–--2022, 2005.

\bibitem[Delbaen and Schachermayer(1994)]{DS}
F.~Delbaen and W.~Schachermayer.
\newblock A general version of the fundamental theorem of asset pricing.
\newblock \emph{Math. Ann.}, 300:\penalty0 463--520, 1994.

\bibitem[Ekeland and Taflin(2005)]{EkelandTaflin2005}
I.~Ekeland and E.~Taflin.
\newblock A theory of bond portfolios.
\newblock \emph{Ann. Appl. Probab.}, 15\penalty0 (2):\penalty0 1260–--1305,
  2005.

\bibitem[Hiriart-Urrut and Lemar\'{e}chal(2004)]{LemHur}
J.-B. Hiriart-Urrut and C.~Lemar\'{e}chal.
\newblock \emph{Fundamentals of Convex Analysis}.
\newblock Springer, 2004.

\bibitem[Jacod and Shiryaev(1980)]{Jahod-Shiryaev}
J.~Jacod and A.~N. Shiryaev.
\newblock \emph{Limit Theorems for Stochastic Processes}.
\newblock Springer, 1980.

\bibitem[Kabanov and Kramkov(1994)]{KabKram1}
Y.~Kabanov and K.~Kramkov.
\newblock Large financial markets: asymptotic arbitrage and contiguity.
\newblock \emph{Probab. Theory Appl.}, 39\penalty0 (1):\penalty0 222–--229,
  1994.

\bibitem[Kabanov and Kramkov(1998)]{KabKram2}
Y.~Kabanov and K.~Kramkov.
\newblock Asymptotic arbitrage in large financial markets.
\newblock \emph{Finance Stoch.}, 2:\penalty0 143--172, 1998.

\bibitem[Kardaras(2013)]{Kardaras_emery}
C.~Kardaras.
\newblock On the closure in the emery topology of semimartingale wealth-process
  sets.
\newblock \emph{Ann. Appl. Probab.}, 23\penalty0 (4):\penalty0 1355--1376,
  2013.

\bibitem[Klein(2000)]{Klein2000}
I.~Klein.
\newblock A fundamental theorem of asset pricing for large financial markets.
\newblock \emph{Math. Finance}, 10:\penalty0 443--458, 2000.

\bibitem[Klein(2003)]{Klein2003}
I.~Klein.
\newblock Free lunch for large financial markets with continuous price
  processes.
\newblock \emph{Ann. Appl. Probab.}, 13\penalty0 (4):\penalty0 1494--1503,
  2003.

\bibitem[Klein(2006)]{Klein2006}
I.~Klein.
\newblock Market free lunch and large financial markets.
\newblock \emph{Ann. Appl. Probab.}, 16\penalty0 (4):\penalty0 2055--2077,
  2006.

\bibitem[Klein and Schachermayer(1996{\natexlab{a}})]{KleinSchach1}
I.~Klein and W.~Schachermayer.
\newblock Asymptotic arbitrage in non-complete large financial markets.
\newblock \emph{Teory Probab. Appl.}, 41\penalty0 (4):\penalty0 927--934,
  1996{\natexlab{a}}.

\bibitem[Klein and Schachermayer(1996{\natexlab{b}})]{KleinSchach2}
I.~Klein and W.~Schachermayer.
\newblock A quantitative and a dual versions of the {H}almos-{S}avage theorem
  with applications to mathematical finance.
\newblock \emph{Ann. Probab.}, 24\penalty0 (2):\penalty0 867--881,
  1996{\natexlab{b}}.

\bibitem[Kramkov and Schachermayer(1999)]{KS}
D.~Kramkov and W.~Schachermayer.
\newblock The asymptotic elasticity of utility functions and optimal investment
  in incomplete markets.
\newblock \emph{Ann. Appl. Probab.}, 9:\penalty0 904--950, 1999.

\bibitem[Mostovyi(2012)]{Mostovyi2011}
O.~Mostovyi.
\newblock Necessary and sufficient conditions in the problem of optimal
  investment with intermediate consumption.
\newblock \emph{arXiv:1107.5852v1 [q-fin.PM], accepted in Finance Stoch.},
  2012.

\bibitem[Ringer and Tehranchi(2006)]{TehranchiRinger}
N.~Ringer and M.~Tehranchi.
\newblock Optimal portfolio choice in the bond market.
\newblock \emph{Finance Stoch.}, 10\penalty0 (4):\penalty0 553–--573, 2006.

\bibitem[Rockafellar(1970)]{Rok}
R.~T. Rockafellar.
\newblock \emph{Convex Analysis}.
\newblock Princeton Univ. Press., 1970.

\bibitem[Strong(2013)]{Strong}
W.~Strong.
\newblock Fundamental theorems of asset pricing for piecewise semimartingales
  of stochastic dimension.
\newblock \emph{Finance Stoch.}, 2013.
\newblock to appear.

\bibitem[Taflin(2005)]{Taflin2005}
E.~Taflin.
\newblock Bond market completeness and attainable contingent claims.
\newblock \emph{Finance Stoch.}, 9:\penalty0 429–--452, 2005.

\bibitem[\v{Z}itkovi\'{c}(2005)]{Zitkovic}
G.~\v{Z}itkovi\'{c}.
\newblock Utility maximization with a stochastic clock and an unbounded random
  endowment.
\newblock \emph{Ann. Appl. Probab.}, 15:\penalty0 748--777, 2005.

\end{thebibliography}
\end{document}